\documentclass [11pt,draftclsnofoot,onecolumn] {IEEEtran}
\usepackage{amsmath, amsthm, amssymb}
\newtheorem{lemma}{Lemma}

\newtheorem{theorem}{Theorem}
\newtheorem{corollary}{Corollary}

\usepackage{color}

\usepackage{rotating,cite}

\newcommand{\bfm}[1]{\mbox{\boldmath $#1$}}

\begin{document}

\title{A Note on the Deletion Channel Capacity}

\author{Mojtaba~Rahmati and Tolga~M.~Duman 

\thanks{This work is funded by the National Science Foundation under the \mbox{contract NSF-TF 0830611}.}

\thanks{M. Rahmati is with the School of Electrical, Computer and Energy Engineering (ECEE) of Arizona State University, Tempe, AZ 85287-5706, USA (email: mojtaba@asu.edu); T. M. Duman is with the Department of Electrical and Electronics Engineering, Bilkent University, Bilkent, Ankara, 06800, Turkey (email: duman@ee.bilkent.edu.tr) and he is on leave from the School of ECEE of Arizona State University.
}

}

\maketitle

\begin{abstract}
Memoryless channels with deletion errors as defined by a stochastic channel matrix allowing for bit drop outs are considered in which transmitted bits are either independently deleted with probability $d$ or unchanged with probability $1-d$. Such channels are information stable, hence their Shannon capacity exists. However, computation of the channel capacity is formidable, and only some upper and lower bounds on the capacity exist. In this paper, we first show a simple result that the parallel concatenation of two different independent deletion channels with deletion probabilities $d_1$ and $d_2$, in which every input bit is either transmitted over the first channel with probability of $\lambda$ or over the second one with probability of $1-\lambda$, is nothing but another deletion channel with deletion probability of $d=\lambda d_1+(1-\lambda)d_2$. We then provide an upper bound on the concatenated deletion channel capacity $C(d)$ in terms of the weighted average of $C(d_1)$, $C(d_2)$ and the parameters of the three channels. An interesting consequence of this bound is that $C(\lambda d_1+(1-\lambda))\leq \lambda C(d_1)$ which enables us to provide an improved upper bound on the capacity of the i.i.d. deletion channels, i.e., $C(d)\leq 0.4143(1-d)$ for $d\geq 0.65$. This generalizes the asymptotic result by Dalai~\cite{Dalai_10} as it remains valid for all $d\geq 0.65$. Using the same approach we are also able to improve upon existing upper bounds on the capacity of the deletion/substitution channel.
\end{abstract}


\section{Introduction}
Channels with synchronization errors can be well modeled using bit drop outs and/or bit insertions as well as random errors. There are many different models adopted in the literature to describe these errors. Among them, a relatively general model is employed by Dobrushin~\cite{dobrushin} where memoryless channels with synchronization errors are described by a channel matrix allowing for the channel outputs to be of different lengths for different uses of the channel. As proved in the same paper, for such channels, information stability holds and Shannon capacity exists. However, the determination of the capacity remains elusive as the mutual information term to be maximized does not admit a single letter or finite letter form.

In the existing literature, several specific instances of this model are more widely studied. For instance, by a proper selection of the stochastic channel transition matrix, one obtains the i.i.d. deletion channel which represents one of the simplest models allowing for bit drop-outs which is the model considered in this paper. In a binary i.i.d. deletion channel, the transmitted bits are either received correctly and in the right order or deleted from the transmitted sequence altogether with a certain probability $d$ independent of each other. Neither the receiver nor the transmitter knows the positions of the deleted bits. Despite the simplicity of the model, the  capacity for this channel is still unknown, and only a few upper and lower bounds are available \cite{mitzenmacher2006,drinea2007improved,drinea2007,dario}. Other special cases of the general model by Dobrushin are the Gallager model allowing for insertions, deletions and substitution errors in which every transmitted bit is either deleted with probability of $d$, replaced with two random bits with probability of $i$, flipped with probability of $f$ or received correctly with probability of $1-d-i-f$. Substituting $i=0$ in the Gallager model results into the deletion/substitution channel model which is also considered in this paper. Another look at the deletion/substitution channel can be as a series concatenation of two independent channels such that the first one is a deletion only channel with deletion probability of $d$ and the second one is binary symmetric channel (BSC) with cross error probability of $s=\frac{f}{1-d}$. There are also some capacity upper and lower bounds for the Gallager's deletion channel model in the literature, e.g., ~\cite{dario2,IT-paper, IT_synch}.

In this paper, we prove that the capacity of an i.i.d. deletion channel with deletion probability of $d$ as an arithmetic mean of two different deletion probabilities $d_1$ and $d_2$, i.e., $d=\lambda d_1+(1-\lambda)d_2$ for $\lambda\in[0,1]$, can be upper bounded in terms of the capacity and the parameters of the two newly considered deletion channels. The proof relies on the simple observation that the deletion channel with deletion probability $d$ can be considered as the parallel concatenation of two independent deletion channels with deletion probabilities $d_1$ and $d_2$ where each bit is either transmitted over the first channel with probability $\lambda$ or the second channel with probability $1-\lambda$. 

Thanks to the presented inequality relation among the deletion channels capacity, we are able to improve upon the existing upper bounds on the capacity of the deletion channel for $d\geq 0.65$ \cite{dario}. The improvement is  the result of the fact that the currently known best upper bounds are not convex for some range of deletion probabilities. More precisely, our result allows us to convexify the existing deletion channel capacity upper bound for $d \geq 0.65$, leading to a significant improvement of the upper bound. In other words, we 
are able to prove that for $0\leq\lambda\leq 1$, $C(\lambda d+1-\lambda)\leq \lambda C(d)$, resulting in $C(d)\leq 0.4143 (1-d)$ for $d\geq 0.65$ which is tighter than the result in~\cite{dario}. The same result for the asymptotic scenario $d\rightarrow 1$ was also obtained in \cite{Dalai_10} using a different approach; however our result is valid for $d\geq 0.65$ hence more general. We also note that the best known limiting lower bound (as $d \rightarrow 1$) is $0.1185 (1-d)$ \cite{mitzenmacher2006}. We also demonstrate that a similar improvement is possible for the case of deletion/substitution channels. As an example, we can prove that for $s = 0.03$, an improved capacity upper bound is obtained for $d \geq 0.6$ over the best existing result given in \cite{dario2}.

The paper is organized as follows. In Section~\ref{Sec:quasi_convexity_proof}, we prove the main result of the paper which relates the capacity of the three different deletion channels through an inequality. In Section~\ref{Sec:generalizations}, we generalize the result to the case of deletion/substitution channels and the parallel concatenation of more than two channels. In Section~\ref{Sec:improved-upper-bounds}, we present tighter upper bounds on the capacity of the deletion and deletion/substitution channels based on previously known best upper bounds, and comment on the limit of the capacity as the deletion probability approaches unity. We conclude the paper in Section~\ref{Sec:conclusions}.

\section{Main Theorem}\label{Sec:quasi_convexity_proof}

In this section, we provide the main result of the paper on the capacity of the deletion channel and its proof. Furthermore, we present a simple proof for the special case with $d_2 = 0$, i.e., $C(\lambda d_1+1-\lambda)\leq \lambda C(d_1)$.

The theorem below states our basic result whose proof hinges on a simple observation.

\begin{theorem}\label{thm_convexity}
Let $C(d)$ denotes the capacity of the i.i.d. deletion channel with deletion probability $d$, $\lambda\in [0,1]$ and $d=\lambda d_1+(1-\lambda)d_2$, then we have
\begin{eqnarray}\label{eq_del-quasiconvexity}
C(d)&\leq & \lambda C(d_1)+(1-\lambda)C(d_2)+(1-d)\log(1-d)\nonumber\\
&&-\lambda (1-d_1)\log(\lambda (1-d_1))-(1-\lambda)(1-d_2)\log((1-\lambda)(1-d_2)).
\end{eqnarray}
\end{theorem}

\begin{proof}
Let us consider two different deletion channels, ${\cal{C}}_1$ and ${\cal{C}}_2$, with deletion probabilities $d_1$ and $d_2$, input sequences of bits $\bfm X_1$ and $\bfm X_2$, and output sequences of bits $\bfm Y_1$ and $\bfm Y_2$, respectively. Denote their Shannon capacities by $C(d_1)$ and $C(d_2)$, respectively. Given a specific $\lambda \in (0,1)$, define a new binary input channel $\cal{C'}$ (shown in Fig.~\ref{fig_ch_model_conv}) with input sequence of bits $\bfm X$ and output sequence of bits $\bfm Y$ as follows: each channel input symbol is transmitted through ${\cal{C}}_1$ with probability $\lambda$, and through ${\cal{C}}_2$ with probability $1-\lambda$, independently of each other. Neither the transmitter nor the receiver knows the specific realization of the ``individual channel selection events,'' i.e., they do not know which specific subchannel a symbol is transmitted through, and which specific subchannel each output symbol is received from. The following two lemmas demonstrate that 1) the new channel is a new i.i.d. deletion channel with deletion probability $d = \lambda d_1 + (1-\lambda) d_2$, 2) if appropriate side information be provided for the transmitter and the receiver then the capacity of the genie-aided channel is upper bounded by $$\lambda C(d_1) + (1-\lambda) C(d_2)+(1-d)\log(1-d)-\lambda(1-d_1) \log(\lambda (1-d_1))-(1-\lambda)(1-d_2)\log((1-\lambda)(1-d_2)).$$ Combining these two results, the proof of the theorem follows easily by noting that the capacity of the new channel $\cal{C'}$ cannot decrease with side information.
\end{proof}
\begin{figure}
   \centering
    \includegraphics[trim=24mm 82mm 40mm 55mm,clip,  width=.75\textwidth]{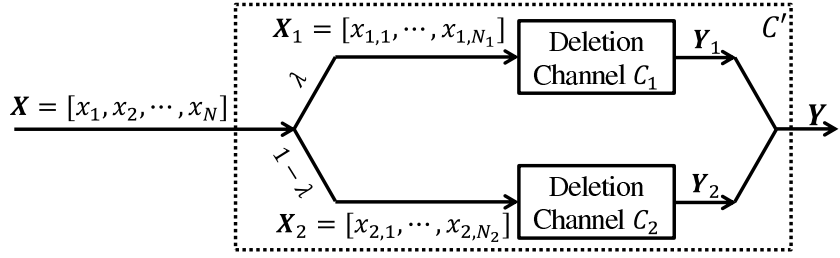}
    \caption{Channel Model $\cal{C'}$}
    \label{fig_ch_model_conv}
\end{figure}

The following two lemmas are employed in the proof of the theorem.
\begin{lemma}\label{lem_parallel}
$\cal{C'}$ as defined in the proof of the theorem above is nothing but a deletion channel with deletion probability $d = \lambda d_1+ (1-\lambda) d_2$.
\end{lemma}
\begin{proof}
For each use of the channel $\cal{C'}$, for any input symbol $x \in \cal{X}$ and channel output $y \in \cal{Y}$, the transition probability is given by $P\{{\cal{C}}_1 \mbox{ is used}\} d_1 + P\{{\cal{C}}_2 \mbox{ is used}\} d_2 = \lambda d_1 + (1-\lambda) d_2$. Noting that the subchannels are memoryless and the channel selection events are independent of each other, this transition matrix precisely defines a deletion channel with deletion probability $d = \lambda d_1 + (1-\lambda) d_2$.
\end{proof}

\begin{lemma}\label{lem_capacity}
The capacity of the channel $\cal{C'}$ as defined in the proof of the theorem above is upper bounded by $$\lambda C(d_1) + (1-\lambda) C(d_2)+(1-d)\log(1-d)-\lambda (1-d_1)\log(\lambda (1-d_1))-(1-\lambda)(1-d_2)\log((1-\lambda)(1-d_2)).$$
\end{lemma}
\begin{proof}
We first define a new genie-aided channel which is obtained by providing the transmitter and the receiver of the channel $\cal{C'}$ with appropriate side information, then derive an upper bound on the capacity of the genie-aided channel which is also an upper bound on the capacity of the channel $\cal{C'}$. More precisely, we provide the transmitter with side information on which channel is being used for each transmitted symbol ($\bfm X=\bfm X_1\bfm X_2$), and the receiver with side information on which channel the received symbol comes from ($\bfm Y=\bfm Y_1\bfm Y_2$), and reveal the side information on the fragmentation information, i.e., random process $\bfm F_y$, to the receiver such that by knowing $\bfm F_y$, $\bfm Y_1$ and $\bfm Y_2$, one can retrieve $\bfm Y$. $\bfm F_y$ is defined as an $M$-tuple $\bfm F_y=(f_y[1],\cdots,f_y[M])$, where $M$ denotes the length of the received sequence $\bfm Y$, i.e., $M=|Y|$, and $f_y[i]\in\{1,2\}$ denotes the index of the channel the $i$-th received bit is coming from. We also define $\bfm F_x$ which determines the fragmentation process from the random process $\bfm X$ to $\bfm X_1$ and $\bfm X_2$ as an $N$-tuple $\bfm F_x=(f_x[1],\cdots,f_x[N])$, where $f_x[i]\in\{1,2\}$ denotes the index of the channel the $i$-th bits is going through.

Since $\bfm X \to (\bfm X_1,\bfm X_2,\bfm F_x)\to (\bfm Y_1,\bfm Y_2,\bfm F_y) \to \bfm Y$ form a Markov chain, we can write
\begin{eqnarray}\label{eq_I}
I(\bfm X;\bfm Y)&\leq & I(\bfm X_1,\bfm X_2,\bfm F_x;\bfm Y_1,\bfm Y_2,\bfm F_y)\nonumber\\
&=&I_1+I_2+I_3,
\end{eqnarray}
where $I_1=I(\bfm X_1,\bfm X_2,\bfm F_x;\bfm Y_1)$, $I_2=I(\bfm X_1,\bfm X_2,\bfm F_x;\bfm Y_2|\bfm Y_1)$ and $I_3=I(\bfm X_1,\bfm X_2,\bfm F_x;\bfm F_y|\bfm Y_1,\bfm Y_2)$. For $I_1$, we have
\begin{eqnarray}\label{eq_I1}
I_1&=&I(\bfm X_1;\bfm Y_1)+I(\bfm X_2,\bfm F_x;\bfm Y_1|\bfm X_1)\nonumber\\
&=&I(\bfm X_1;\bfm Y_1),
\end{eqnarray}
where we used the fact that $P(\bfm Y_1|\bfm X_1,\bfm X_2,\bfm F_x)=P(\bfm Y_1|\bfm X_1)$, i.e., $\bfm Y_1$ is independent of $\bfm X_2$ and $\bfm F_x$ conditioned on $\bfm X_1$. Furthermore, by using the facts that $P(\bfm Y_2|\bfm X_2,\bfm Y_1)=P(\bfm Y_2|\bfm X_2)$ and $P(\bfm Y_2|\bfm X_1,\bfm X_2, \bfm F_x, \bfm Y_1)=P(\bfm Y_2|\bfm X_2)$, we obtain
\begin{eqnarray}\label{eq_I2}
I_2&=&I(\bfm X_2;\bfm Y_2|\bfm Y_1)+I(\bfm X_1,\bfm F_x;\bfm Y_2|\bfm Y_1, \bfm X_2)\nonumber\\
&=&H(\bfm Y_2|\bfm Y_1)-H(\bfm Y_2|\bfm X_2)\nonumber\\
&\leq& I(\bfm X_2;\bfm Y_2).
\end{eqnarray}

We are not able to derive the exact value of $I_3$, therefore we derive an upper bound on $I_3$ which results in an upper bound on $I(\bfm X,\bfm Y)$. For $I_3$, if we define $N_i=|X_i|$ and $M_i=|Y_i|$ as the length of the transmitted and received sequences form the $i$-th channel, respectively, then we can write
\begin{eqnarray}
I_3&=& H(\bfm F_y|\bfm Y_1,\bfm Y_2)-H(\bfm F_y|\bfm Y_1,\bfm Y_2,\bfm X_1,\bfm X_2,\bfm F_x)\nonumber\\
&\leq & H(\bfm F_y|\bfm Y_1,\bfm Y_2)\nonumber\\
&=& H(\bfm F_y|\bfm M_1,\bfm M_2).
\end{eqnarray}
\color{black}
For fixed $M_1$ and $M_2$, there are ${{M_1+M_2}\choose M_2}$ possibilities for $\bfm F_y=(f_y[0],\cdots,f_y[{M_1}])$. Therefore, we obtain 
\begin{align}
H(\bfm F_y|\bfm M_1=M_1, \bfm M_2=M_2)&\leq \log\left({{M_1+M_2}\choose M_2} \right)\nonumber\\
&\leq (M_1+M_2)\log{(M_1+M_2)}-M_1\log(M_1)-M_2\log(M_2),
\end{align}
where we have used the inequality $\log{n\choose k}\leq nH_b(\frac{k}{n})$ provided in \cite[p.~353]{cover}. Due to the fact that $(x+a)\log(x+a)-x\log(x)$ is a concave function of $x$ for $a>0$, and $E\{\bfm M_1|\bfm M_2=M_2\}=(N-M_2)\frac{\lambda(1-d_1)}{\lambda+(1-\lambda) d_2}$ (see Appendix~\ref{app-M1M2}), by applying Jensen's inequality, we can write
\begin{eqnarray}
I_3&\leq &E_{\bfm M_1,\bfm M_2}\{H(\bfm F_y| \bfm M_1,\bfm M_2)\}\nonumber\\
&\leq &E_{\bfm M_2}\bigg\{(E\{\bfm M_1|\bfm M_2\}+\bfm M_2)\log(E\{\bfm M_1|\bfm M_2\}+\bfm M_2)\nonumber\\
&&\quad \quad \quad -E\{\bfm M_1|\bfm M_2\}\log(E\{\bfm M_1|\bfm M_2\})-\bfm M_2\log(\bfm M_2)\bigg\}\nonumber\\
&=&E\bigg\{ \left(\frac{\lambda(N-\bfm M_2)(1-d_1)}{\lambda+(1-\lambda) d_2}+\bfm M_2\right)\log\left(\frac{\lambda(N-\bfm M_2)(1-d_1)}{\lambda+(1-\lambda) d_2}+\bfm M_2\right)\nonumber\\
&&\quad -\frac{\lambda(N-\bfm M_2)(1-d_1)}{\lambda+(1-\lambda) d_2}\log\left(\frac{\lambda(N-\bfm M_2)(1-d_1)}{\lambda+(1-\lambda) d_2}\right)-\bfm M_2\log(\bfm M_2)\bigg\}.
\end{eqnarray}
Furthermore since $(a(b-x)+x)\log(a(b-x)+x)-a(b-x)\log(a(b-x))-x\log(x)$ is a concave function of $x$ for $a>0$ and $0<x\leq b$, and $E\{\bfm M_2\}=N(1-\lambda)(1-d_2)$ (see Appendix~\ref{app-M1M2}), by applying Jensen's inequality, we obtain
\begin{eqnarray}\label{eq_I3}
I_3&\leq &N(\lambda (1-d_1)+(1-\lambda) (1-d_2))\log(N(\lambda (1-d_1)+(1-\lambda) (1-d_2)))\nonumber\\
&&-N\lambda(1-d_1)\log(N\lambda (1-d_1))-N(1-\lambda)(1-d_2)\log(N(1-\lambda)(1-d_2))\nonumber\\
&=&N(\lambda (1-d_1)+(1-\lambda) (1-d_2))\log(\lambda (1-d_1)+(1-\lambda) (1-d_2))\nonumber\\
&&-N\lambda (1-d_1)\log(\lambda (1-d_1))-N(1-\lambda)(1-d_2)\log((1-\lambda)(1-d_2)).
\end{eqnarray}

On the other hand, for $I(\bfm X_i;\bfm Y_i)$ ($i\in\{1,2\}$), we can write
\begin{eqnarray}\label{eq_I(X_i;Y_i)}
I(\bfm X_i;\bfm Y_i)&=&I(\bfm X_i;\bfm Y_i,\bfm N_i)-I(\bfm X_i;\bfm N_i|\bfm Y_i)\nonumber\\
&=&I(\bfm X_i;\bfm Y_i|\bfm N_i)+I(\bfm X_i;\bfm N_i)-I(\bfm X_i;\bfm N_i|\bfm Y_i)\nonumber\\
&\leq& I(\bfm X_i;\bfm Y_i|\bfm N_i)+H(\bfm N_i)\nonumber\\
&\leq &I(\bfm X_i;\bfm Y_i|\bfm N_i)+\log(N+1)\nonumber\\
&=&\sum_{N_i=0}^{N}P(\bfm N_i=N_i)I(\bfm X_i;\bfm Y_i|\bfm N_i=N_i)+\log(N+1),
\end{eqnarray}
where in deriving the first inequality we have used the facts that $H(\bfm N_i|\bfm X_i)=0$ and \mbox{$I(\bfm X_i;\bfm N_i|\bfm Y_i)\geq 0$}, and in deriving the second equality the fact that
\begin{equation}
H(\bfm N_i)=-\sum_{n=0}^{N}{N\choose n} \lambda^{n} (1-\lambda)^{N-n}\log\left({N\choose n} \lambda^{n} (1-\lambda)^{N-n}\right)\leq \log(N+1).
\end{equation}
Furthermore, as it is shown in~\cite{dario}, for a finite length transmission over the deletion channel, the mutual information rate between the transmitted and received sequences can be upper bounded in terms of the capacity of the channel after adding some appropriate term, which can be spelled out as \cite[Eqn.~(39)]{dario}
\begin{equation}\label{eq_dario}
I(\bfm X_i;\bfm Y_i|\bfm N_i=N_i)\leq N_iC(d_i)+H(\bfm D_i|\bfm N_i=N_i),
\end{equation} 
where $\bfm D_i$ denotes the number of deletion through the transmission of $N_i$ bits over the $i$-th channel and $$H(\bfm D_i|\bfm N_i=N_i)=-\sum_{n=0}^{N_i}{N_i\choose n}d_i^n(1-d_i)^{N_i-n}\log\left({N_i\choose n}d_i^n(1-d_i)^{N_i-n}\right)\leq \log{(N_i+1)}.$$
Substituting~\eqref{eq_dario} into~\eqref{eq_I(X_i;Y_i)}, we have
\begin{eqnarray}\label{eq_I(X_i;Y_i)_2}
I(\bfm X_i;\bfm Y_i)&\leq & \sum_{N_i=0}^{N} P(\bfm N_i=N_i)\left(N_i C(d_i)+\log(N_i+1)\right)+\log(N+1)\nonumber\\
&\leq & \lambda_i N C(d_i)+\log(\lambda_i N +1)+\log(N+1),
\end{eqnarray}
where the last inequality results since $\log(x)$ is a concave function of $x$, and $\lambda_1=\lambda$ and $\lambda_2=1-\lambda$. Finally, by substituting~\eqref{eq_I(X_i;Y_i)_2}, \eqref{eq_I3}, \eqref{eq_I2} and \eqref{eq_I1} in~\eqref{eq_I}, we obtain
\begin{eqnarray}
I(\bfm X;\bfm Y)&\leq &N\lambda C(d_1)+ \log(\lambda N+1)+N(1-\lambda)C(d_2)+\log((1-\lambda) N+1)\nonumber\\
&&+\ 2\log(N+1)+N(1-d)\log(1-d)-N\lambda (1-d_1)\log(\lambda (1-d_1))\nonumber\\
&&-N(1-\lambda)(1-d_2)\log((1-\lambda)(1-d_2)).\nonumber
\end{eqnarray}
By dividing both sides of the above inequality by $N$, letting $N$ go to infinity, and noting that the inequality is valid for any input distribution $P(\bfm X)$, the proof follows. 
\end{proof}

Note that for the special case of ${\cal{C}}_2$ being a pure deletion channel, i.e., $d_2=1$, the presented upper bound~\eqref{eq_del-sub-quasiconvexity} results into $C(\lambda d_1+1-\lambda)\leq \lambda C(d_1)$. One can observe that to prove the relation $C(\lambda d_1+1-\lambda)\leq \lambda C(d_1)$, there is no need for the entire proof given in Lemma~\ref{lem_capacity}. More precisely, when ${\cal{C}}_2$ is a pure deletion channel, $\bfm X\to \bfm X_1\to\bfm Y_1\to \bfm Y$ form a Markov chain ($\bfm Y=\bfm Y_1$), therefore we can write
\begin{align}
I(\bfm X;\bfm Y)&\leq I(\bfm X_1;\bfm Y_1)\nonumber\\
&\leq \lambda N C(d_1)+\log(\lambda_1 N +1)+\log(N+1),
\end{align}
where the last inequality holds due to \eqref{eq_I(X_i;Y_i)_2}. Furthermore, by dividing both sides of the above inequality by $N$, letting $N$ go to infinity, and the fact that the inequality is valid for any input distribution $P(\bfm X)$, we arrive at $C(\lambda d_1+1-\lambda)\leq \lambda C(d_1)$. 

Another observation from the result $C(\lambda d_1+(1-\lambda))\leq \lambda C(d_1)$ is that by series concatenation of two independent deletion channels with deletion probabilities $d_1$ and $1-\lambda$, we also arrive at a deletion channel with deletion probability of $d=\lambda d_1+1-\lambda$. Therefore we can say that the capacity of the series concatenation of two independent deletion channels can be upper bounded in terms of the capacity of one of them and the parameters of the other.

\section{Some Generalizations and Implications}\label{Sec:generalizations}

\subsection{Generalization to the Case of Deletion/Substitution Channel }

In a deletion/substitution channel (special case of the Gallager channel model without any insertions) with parameters ($d$,$f$), any transmitted bit is either deleted with probability of $d$ or flipped with probability of $f$ or received correctly with probability of $1-d-f$, where neither the transmitter nor the receiver have any information about the position of the deleted and flipped bits. It is easy to show that the result of Theorem~\ref{thm_convexity} can also be generalized to the deletion/substitution channel as given in the following corollary. 
\begin{corollary}
Let $C(d,f)$ denotes the capacity of the deletion/substitution channel with deletion probability $d$ and flip probability $f$,  $\lambda\in [0,1]$, $d=\lambda d_1+(1-\lambda)d_2$ and $f=\lambda f_1+(1-\lambda)f_2$, then we have
\begin{eqnarray}
C(d,f)&\leq & \lambda C(d_1,f_1)+(1-\lambda)C(d_2,f_2)+(1-d)\log(1-d)\nonumber\\
&&-\lambda (1-d_1)\log(\lambda (1-d_1))-(1-\lambda)(1-d_2)\log((1-\lambda)(1-d_2)).
\end{eqnarray}
\end{corollary}
\begin{proof}
The proof of Lemma~\ref{lem_parallel} simply holds if we consider ${\cal{C}}_1$ in Fig.~\ref{fig_ch_model_conv} as a deletion/substitution channel with parameters ($d_1$,$f_1$) and ${\cal{C}}_2$ as another deletion/substitution channel with parameters ($d_2$,$f_2$), then ${\cal{C}}$ becomes also a deletion/substitution channel with parameters $(\lambda d_1+(1-\lambda)d_2,\lambda f_1+(1-\lambda)f_2)$. Furthermore, replacing the deletion channel ${\cal{C}}_i$ with deletion probability $d_i$ with a deletion/substitution channel with parameters ($d_i$,$f_i$) does not change the distribution of $\bfm N_i$ and $\bfm M_i$. Therefore, the proof of Lemma~\ref{lem_capacity} holds for the deletion/substitution channel as well.
\end{proof}
Note that a deletion/substitution channel with parameters ($d,f$) can be considered as a series concatenation of two independent channels where the first one is a deletion only channel with deletion probability of $d$ and the second one is a binary symmetric channel (BSC) with cross error probability $s=\frac{f}{1-d}$ ($1-d-f\leq 1$ and if $d=1$ then $s=0$). If we define $C_s(d,s)=C(d,(1-d)s)$, then for $d_2=1$ and $f_2=0$, we obtain
\begin{align}\label{eq_del-sub-quasiconvexity}
C_s(\lambda d_1+1-\lambda ,s)\leq & \lambda C_s(d_1,s).
\end{align}

\subsection{Parallel Concatenation of More Than Two Channels}
So far, we considered the parallel concatenation of two independent deletion channels which is useful in improving upon the existing upper bounds. However, we can also consider the parallel concatenation of more than two deletion channels. If we define the deletion channel $\cal{C}$ as a parallel concatenation of $P$ independent deletion channels ${\cal{C}}_p$ with deletion probability $d_p$ ($p=\{1,\cdots,P\}$) where each input bit is transmitted with probability $\lambda_p$ over ${\cal{C}}_p$, and modify the definition of $\bfm F_y$ such that $f_y[i]\in\{1,\cdots,P\}$ denotes the index of the channel the $i$-th bit is coming from, then for $d=\sum_{p=1}^P \lambda_pd_p$, we have
\begin{align}
C(d)\leq \sum_{p=1}^P \lambda_p C(d_p)+ (1-d)\log(1-d)-\sum_{p=1}^P \lambda_p (1-d_p)\log(\lambda_p(1-d_p)),
\end{align}
where $\sum_{p=1}^P\lambda_p=1$. Note, however, that this result does not give any tighter upper bounds on the deletion channel capacity than the one obtained by considering the parallel concatenation of only two independent deletion channels. 

\section{Improved Upper Bounds on the Deletion Channel  Capacity}\label{Sec:improved-upper-bounds}

An interesting application of the result~\eqref{eq_del-quasiconvexity} on the capacity of the deletion and deletion/substitution channels is in obtaining improved capacity upper bounds. For instance, the best known upper bound on the deletion channel capacity is not convex for $d \geq 0.65$ as shown in Fig.~\ref{del_cap_upper_bound} (with values taken from the boldfaced values in Table IV of~\cite{dario}). As clarified in the table, the best known values for small $d$ are due to \cite{diggavi-capacity}, for a wide range (up to $d \sim .8$) are due to the ``fourth version'' of the upper bound (named $C_4$ in~\cite{dario}), and for large values of $d$ are due to the ``second version" named $C^*_2$ in the same paper. Therefore, the deletion channel capacity upper bound can be improved for $d \in (0.65,1)$ as $C(1-0.35\lambda)\leq \lambda C(0.65)\leq \lambda C_{4}(0.65)$ with $0 \leq \lambda \leq 1$. That is, we have  $C(d) \leq 0.4143 (1-d)$ for $d\in(0.65,1)$. This is illustrated in Fig.~\ref{del_new_upper_bound}.

\begin{figure}[hbt]
\centerline{
  \hbox{
    \resizebox{110mm}{!}
    {\includegraphics[trim=0mm 6mm 0mm 8mm,clip,  width=1\textwidth]{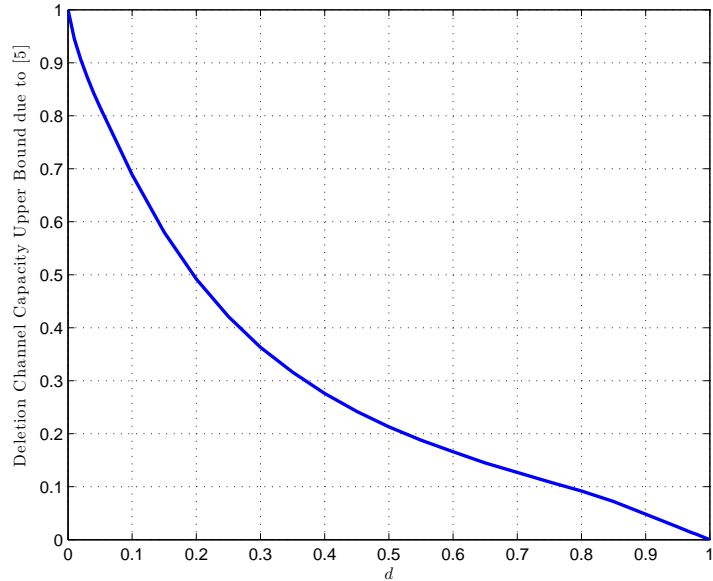}}
  }
}
\caption{Previously best known upper bound on the i.i.d. deletion channel capacity.}
\label{del_cap_upper_bound}
\end{figure}

\begin{figure}[hbt]
\centerline{
  \hbox{
    \resizebox{110mm}{!}
    {\includegraphics[trim=0mm 6mm 0mm 8mm,clip,  width=1\textwidth]{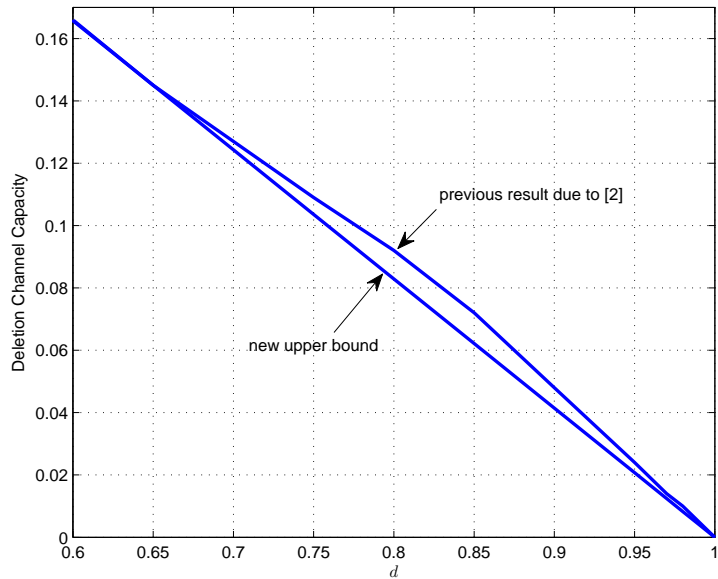}}
  }
}
\caption{Improved upper bound on the deletion channel capacity employing $C(\lambda d+1-\lambda)\leq \lambda C(d)$.}
\label{del_new_upper_bound}
\end{figure}

We note that our result is a generalization of the one in~\cite{Dalai_10} where it was shown that $C(d) \leq 0.4143 (1-d)$ as $d \rightarrow 1$. We also note an earlier asymptotic result on a lower bound derived in~\cite{mitzenmacher2006} which states that $C(d)$ as $d \rightarrow 1$ is larger than $0.1185 (1-d)$.

\begin{figure}[hbt]
\centerline{
  \hbox{
    \resizebox{110mm}{!}
    {\includegraphics[trim=0mm 6mm 0mm 5mm,clip,  width=1\textwidth]{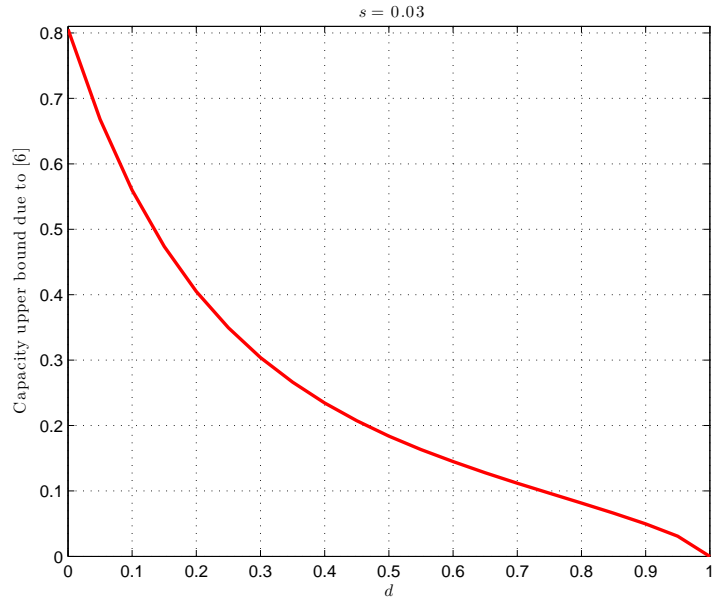}}
  }
}
\caption{Previously best known upper bound on the deletion/substitution channel capacity for $s=0.03$.}
\label{del-sub_cap_upper_bound}
\end{figure}

\begin{figure}[hbt]
\centerline{
  \hbox{
    \resizebox{110mm}{!}
    {\includegraphics[trim=0mm 6mm 0mm 5mm,clip,  width=1\textwidth]{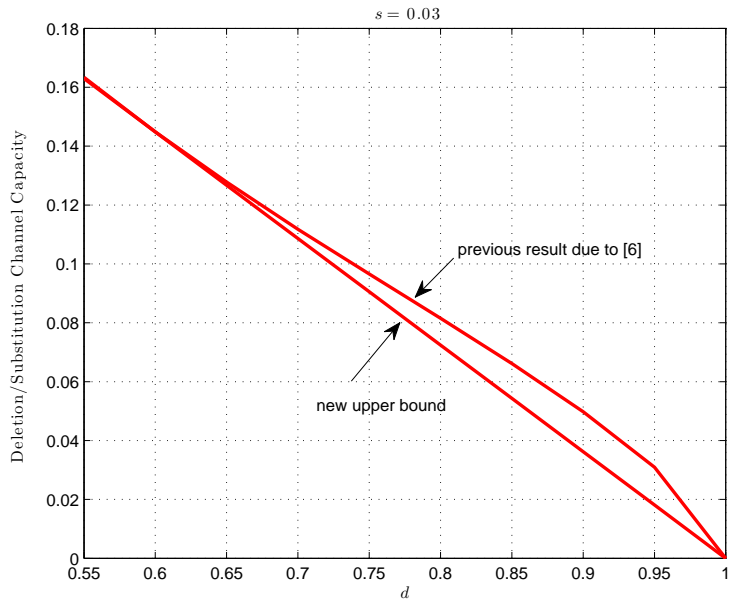}}
  }
}
\caption{Improved upper bound on the deletion/substitution channel capacity for $s=0.03$.}
\label{del-sub_new_upper_bound}
\end{figure}

As another application of the inequality derived in this paper, we can consider the capacity of the deletion/substitution channel. The best known capacity upper bound for this case is given in~\cite{dario2}, e.g., Fig.~1 of \cite{dario2} presents several upper bounds for fixed $s=0.03$ (see~Fig.~\ref{del-sub_cap_upper_bound}). It is clear that this bound is not a convex function of the deletion probability  for $d \geq 0.6$, hence it can be improved. That is, applying the result in our paper, we obtain, for instance for $s=0.03$, $C_s(d,0.03) \leq 0.3621(1-d)$ for $d\geq 0.6$ which is a tighter bound as illustrated in Fig.~\ref{del-sub_new_upper_bound}.

\section{Conclusions}\label{Sec:conclusions}

In this paper, an inequality relating the capacity of a deletion channel to two other deletion channels is found. The main idea is to consider parallel concatenation of two different independent deletion channels and relate the capacity of the resulting deletion channel with the capacity of the first two. An immediate application of this result is in obtaining improved upper bounds on the capacity of the deletion channel as the best available upper bounds are not convex in the deletion probability, and the derived inequality results in a tighter capacity characterization. For an i.i.d. deletion channel, we proved that $C(d) \geq 0.4143 (1-d)$ for all $d \geq 0.65$. This is a stonger result than the earlier characterization in~\cite{Dalai_10} which is valid only asymptotically as $d \rightarrow 1$. We also noted a generalization of the result to the case of a deletion/substitution channel and provided a tigher capacity upper bound for this case as well. \\

\section*{ACKNOWLEDGMENTS}
The authors would like to thank Marco Dalai for his insightful comments on the paper. 

\appendices
\section{Stochastic Properties of $\bfm M_1$ and $\bfm M_2$}\label{app-M1M2}
For $P(\bfm M_1,\bfm M_2)$, we can write
\begin{align}
P(\bfm M_1=M_1,\bfm M_2=M_2)&=\sum_{N_1=M_1}^{N-M_2} P(\bfm M_1=M_1,\bfm M_2=M_2|\bfm N_1=N_1)P(\bfm N_1=N_1)\nonumber\\
&=\sum_{N_1=M_1}^{N-M_2} P(\bfm M_1=M_1|\bfm N_1=N_1)P(\bfm M_2=M_2|\bfm N_1=N_1)P(\bfm N_1=N_1)\nonumber\\
&=\sum_{N_1=M_1}^{N-M_2} {{N_1}\choose {M_1}} d_1^{N_1-M_1}(1-d_1)^{M_1}{{N-N_1}\choose {M_2}} d_2^{N-N_1-M_2}(1-d_2)^{M_2}\times \nonumber\\
&\quad \quad \quad \quad \times {N\choose {N_1}}\lambda^{N_1}(1-\lambda)^{N-N_1}\nonumber\\
&={{N-M_2}\choose {M_1}} {N\choose {M_2}}(\lambda(1-d_1))^{M_1}((1-\lambda)(1-d_2))^{M_2}\times\nonumber\\
&\quad \times \sum_{N_1=M_1}^{N-M_2} {{N-(M_1+M_2)}\choose {N_1-M_1}} (\lambda d_1)^{N_1-M_1} ((1-\lambda)d_2)^{N-N_1-M_2}\nonumber\\
&={{N-M_2}\choose {M_1}} {N\choose {M_2}}(\lambda(1-d_1))^{M_1}((1-\lambda)(1-d_2))^{M_2} d^{N-M_1-M_2}.
\end{align}
Furthermore, due to the structure of the channel $\cal{C}'$, $\bfm M_2$ is binomially distributed, i.e., $P(\bfm M_2=M_2)={N\choose {M_2}}((1-\lambda)(1-d_2))^{M_2} (\lambda+(1-\lambda)d_2)^{N-M_2}$, and as a result $E\{\bfm M_2\}=N(1-\lambda)(1-d_2)$. On the other hand, to obtain $E_{\bfm M_1}\{\bfm M_1|\bfm M_2\}$, we first need to obtain $P(\bfm M_1|\bfm M_2)$, for which we can write
\begin{align}
P(\bfm M_1=M_1|\bfm M_2=M_2)&=\frac{P(\bfm M_1,\bfm M_2)}{P(\bfm M_2)}\nonumber\\
&= {{N-M_2}\choose {M_1}} (\lambda(1-d_1))^{M_1} (\lambda d_1+(1-\lambda)d_2)^{N-M_1-M_2}(\lambda+(1-\lambda)d_2)^{M_2-N}.\nonumber
\end{align}
Therefore, we obtain
\begin{align}
E_{\bfm M_1}\{\bfm M_1|\bfm M_2\}&=\sum_{M_1=0}^{N-M_2}M_1{{N-M_2}\choose {M_1}} (\lambda(1-d_1))^{M_1} (\lambda d_1+(1-\lambda)d_2)^{N-M_1-M_2}(\lambda+(1-\lambda)d_2)^{M_2-N}\nonumber\\
&=(N-M_2)\frac{\lambda(1-d_1)}{\lambda+(1-\lambda)d_2}.
\end{align}

\bibliographystyle{IEEEtran}
\bibliography{myrefs}

\end{document}